\documentclass[11pt,a4paper,final]{article}
\usepackage[latin1]{inputenc}
\usepackage{amsfonts}
\usepackage{amscd}
\usepackage{amsmath}
\usepackage{theorem}
\usepackage{mathrsfs}
\usepackage[dvips]{graphicx}
\usepackage{fancybox,amssymb}
\usepackage{geometry}
\usepackage[english]{babel}
\usepackage{indentfirst}

\newcommand{\R}{\mathbb{R}}
\newcommand{\bD}{\mathbf{D}}

\newcommand{\mP}{\mathbb{P}}
\newcommand{\mE}{\mathbb{E}}
\newcommand{\md}{\,{\rm d}}

\newcommand{\w}{\wedge}
\newcommand{\one}{1\mkern-5mu{\hbox{\rm I}}}

\theoremstyle{break}
\theorembodyfont{}
\newtheorem{Def}{Definition}[section]
\newtheorem{Bem}[Def]{Remark}

\newtheorem{Prop}[Def]{Proposition}

\newtheorem{Bsp}[Def]{Example}

\makeatletter
\newenvironment{proof}{\noindent{\textit{Proof:}}}{%
\unskip\nobreak\hfil\penalty50\hskip1em\null\nobreak
$\Box$
\parfillskip=\z@\finalhyphendemerits=0\endgraf\bigskip}

\let\oldendBsp\endBsp
\def\endBsp{\unskip\nobreak\hfil\penalty50\hskip1em\null\nobreak\hfil%
$\blacksquare$\parfillskip=\z@\finalhyphendemerits=0\endgraf\oldendBsp}
\let\oldendBem\endBem
\def\endBem{\unskip\nobreak\hfil\penalty50\hskip1em\null\nobreak\hfil%
$\blacksquare$\parfillskip=\z@\finalhyphendemerits=0\endgraf\oldendBem}
\makeatother

\author{\small\sc  Julia Eisenberg\footnote{email: jeisenbe@fam.tuwien.ac.at} and Paul Kr\"uhner\smallskip\\\footnotesize Institute of Mathematical Methods in Economics, Vienna University of Technology.}
\date{}
\title{A Note on the Optimal Dividends Paid in a Foreign Currency}

\begin{document} 
\maketitle 
\begin{abstract}\noindent
We consider an insurance entity endowed with an initial capital and a surplus process modelled as a Brownian motion with drift. It is assumed that the company seeks to maximise the cumulated value of expected discounted dividends, which are declared or paid in a foreign currency. The currency fluctuation is modelled as a L\'evy process. We consider both cases: restricted and unrestricted dividend payments. It turns out that the value function and the optimal strategy can be calculated explicitly.

\vspace{6pt}
\settowidth\labelwidth{{\it 2010 Mathematical Subject Classification: }}%
                \par\noindent {\it 2010 Mathematical Subject Classification: }%
                \rlap{Primary}\phantom{Secondary}
                93B05\newline\null\hskip\labelwidth
                Secondary 49L20

\end{abstract}
\section{Introduction}
In the time of low interest rates, investors seek for other sources of income, e.g.\ dividend payments. Indeed, a dividend is a payout that a company can make to its stockholders. Since, the dividends typically reflect the company's earnings, the return on investment can exceed the riskless yields available on the market considerably. 

Usually, investors seek to diversify their risks and invest in companies abroad. On the other hand, large companies expand to other countries and offer their shareholders the possibility to pay out the dividends in the local currency. This exposes investors to the additional risk of currency fluctuations. 

In order to show, how an investor can be affected by exchange rate movements, let us consider the following example. BP plc is one of the world's biggest oil and gas companies, and therefore one of the most popular companies to invest in. On the company's website (www.bp.com) one finds the following information: ``BP declares the dividend in US dollars. BP ordinary shareholders will receive the dividend in sterling and the amount they receive each quarter may vary as a result of changing foreign exchange rates.''
Since December 2014, the quarterly dividend payments of BP for ordinary shareholders equals $0.10$\textdollar{} per share. However, the actual payments in pound sterling amounted for example to $0.066342$\textsterling{} on the 18.12.2015 and to $0.063769$\textsterling{} on the 19.12.2014 (confer www.bp.com/content/dam/bp/pdf/investors/bp-cash-dividends-ordinary-shareholders.pdf). An increase of more than $4\%$ is due to the strengthening of the American dollar.
\\And vice versa, if the company's earnings/dividends are growing or at least not decreasing but its domestic currency (in our case pound sterling) increases versus dollar, the dividend payments become less worth than in US. For instance, in 2009 the quarterly dividend payments of BP for ordinary shareholders equalled $0.14$\textdollar{}/share (confer www.nasdaq.com/symbol/bp/dividend-history), which corresponded to $0.09584$\textsterling{} on the 08.06.2009 and to $0.08503$\textsterling{} on the 08.09.2009, showing a decrease of more than $11\%$ just within a 3-month period.   
\\To mention an example from the insurance/reinsurance industry: Swiss Re declares its regular dividend, paid annually, in CHF. That is, the dividend income of an investor from EU would depend on the exchange rate CHF/EU like described above.

The dividend maximisation problem with different constraints has a long history. Shreve, Lehoczky and Gaver \cite{shreve} maximised the expected discounted dividends up to ruin minus some penalty for a general diffusion process as a surplus process of the considered insurance entity. Asmussen and Taksar \cite{astak} found an explicit solution for a Brownian motion with drift as a surplus.
Hubalek and Schachermayer \cite{hubschach} looked at the problem under a utility function applied on the dividend rates, Grandits et al. \cite{ghsz} maximised the expected utility of the present value of dividends. Since this paper, is not supposed to be a survey on the topic, we refer for example to Albrecher and Thonhauser \cite{AlbThReview} or to Schmidli \cite{HS} and references therein.

Consider now an insurance company seeking to maximise the expected discounted dividends, whereas the dividends are declared or paid in a foreign currency. The surplus of the considered company is assumed to follow a Brownian motion with drift, and the factor describing the currency fluctuations is given by a geometric L\'evy process, which is the basic set-up for practical applications, confer for instance \cite[p. 452]{brigo}.
We consider the cases where the dividend payments are restricted or unrestricted and also allow for dependency between the surplus process and currency movements.

To the best of our knowledge, the problem of dividend maximisation if the dividends are paid in a foreign currency has not been considered before. Usually, one assumes that the discounting/preference rate is a positive deterministic constant. In \cite{eis}, the dividends were maximised under the assumption of a geometric Brownian motion as a discounting factor. Since, a Brownian motion is a L\'evy process, the model considered in the present manuscript is a generalisation of the one considered in \cite{eis}.
Similar to \cite{eis}, we are able to calculate the value function and the optimal strategy explicitly. The results are illustrated by examples.

\section{The Value Function and the Optimal Strategy}
\noindent
Consider an insurance company whose surplus is given by a Brownian motion with drift $X_t=x+\mu t+\sigma W_t$, where $\{W_t\}$ is a standard Brownian motion $\mu,\sigma>0$. The insurance company is allowed to pay out dividends at any amount up to the current surplus at any time, where the accumulated dividends until $t$ are given by $\bD_t$, yielding for the ex-dividend surplus $X^{\bD}$:
\[
X_t^\bD=x+\mu t+\sigma W_t-\bD_t\;.
\]
The consideration will be stopped at the ruin time of $X^\bD$, i.e.\ at time $\tau^\bD:=\inf\{t\geq0: X^\bD(t)= 0\}$. We merely assume that the underlying filtration $\{\mathcal F_t\}$ is complete, right-continuous and that $W$ is an $\mathcal F$-Brownian motion. In this paper we consider too sets of admissible dividend strategies. 
At first, we allow just for dividend strategies $\mathbf {D_t}=\int_0^t u_s\md s$ with $u$ progressively measurable and $0\le u_t\le \xi$ for some boundary $\xi>0$. Moreover, we denote the class of those dividend strategies by $\mathcal A_r$ and we sometimes write $\bD=\{u_t\}$ to indicate that $\bD_t=\int_0^t u_s\md s$.
Later, we allow for unrestricted dividend pays at any time, i.e.\ $D$ is allowed to be any right-continuous increasing and adapted process with $\Delta D_t \leq X_t-D_{t-}$ and we denote the class of those processes by $\mathcal A_u$.

As a risk measure, we consider the value of expected discounted dividends, where the dividends are discounted by a constant preference rate $\delta\in\mathbb R$, i.e. the discounting factor is given by $e^{-\delta t}$ for $t\ge 0$. Since we assume, that the considered insurance company generates its income in a foreign currency but pays dividends in its home currency we need to model the exchange rate as well. The price of one unit of the foreign currency is modelled by a geometric L\'{e}vy process, i.e.\ the exponential of a L\'{e}vy process $L$ starting in some level $l\in\mathbb R$, cf.\ \cite{sato.99}. As usual we assume that $L$ is an $\mathcal F$-L\'evy process. Moreover, for the moment we assume that the L\'evy process $L$ is independent of the Brownian motion $W$ driving the surplus process $X$; for dependencies confer further Remark \ref{r:depndency}. 
We define the return function corresponding to some admissible strategy $\bD=\{u_t\}$ to be
\[
V^{\bD}(l,x):=\mE\Big[\int_0^{\tau^\bD} e^{-\delta t-L_t}u_t \md t\Big],
\]
for $(l,x)\in\R\times\R_+$. Since an investor seeks to maximise the return he tries to find an admissible strategy $\bD^*$ such that
\begin{align}
  V^{\bD^*}(l,x) &= \sup_{\bD\in\mathcal A_r} V^\bD(l,x). \label{e:control problem}
\end{align}

Note that for an arbitrary strategy $\bD$ we have
\begin{equation}
\begin{split}
  V^\bD(l,x) &= \mE\Big[\int_0^{\tau^\bD} e^{-\delta t-L_t}u_t \md t\Big]
           = e^{-l} \mE\Big[\int_0^{\tau^\bD} e^{-\delta t-(L_t-l)}u_t \md t\Big]
           = e^{-l} V^\bD(0,x) 
\end{split}
\label{e:l factors}
\end{equation}
which greatly simplifies the dependency on the current level $l\in\mathbb R$ of the log-exchange rate.

Let us introduce some further notation, namely, let $(A,\nu,\gamma)$ be the L\'evy-Khintchine triplet associated with the L\'evy process $L$, cf.\ \cite[Definition 8.2]{sato.99}. We also define an {\em artificial preference rate}
 $$ \beta := \delta - \frac{A}{2} - \int_{\mathbb R} (e^{-h} - 1+ h1_{\{|h|\leq 1\}}) \nu(dh) + \gamma $$
in case that $\int_{-\infty}^{-1} e^{-x} \nu(dx) < \infty$ and $\beta := -\infty$ otherwise. Note that the control problem \eqref{e:control problem} is well posed if and only if $\beta>0$. Indeed, assume $-\infty <\beta\le 0$, choose $\tilde \bD=\{\frac\mu 2\}$ and denote by $\tau^{\tilde \bD}$ the ruin time of the surplus process $X_t^{\tilde \bD}$ under the dividend strategy $\tilde \bD$. Then, confer for instance Borodin and Salminen \cite[p. 295]{bs}, one has $\mP\big[\tau^{\tilde \bD} =\infty\big]=1-e^{-\mu x}$,
implying for $x>0$ $$V(x)\ge V^{\tilde \bD}(x)=\frac \mu2\mE\Big[\int_0^{\tau^{\tilde \bD}}e^{-\beta t}\md t\Big]=\infty,$$ 
which shows that the control problem \eqref{e:control problem} is not well-posed if $\beta\leq 0$. The case $\beta=-\infty$ can be treated similarly. Thus, in the following we assume $\beta>0$.
\\The HJB equation corresponding to the problem is given by
\begin{equation*}
\begin{split}
&\frac A2 V_{ll}(l,x) + \int \left(V(l+h,x) - V(l,x) - h1_{\{|h|\leq 1\}}V_l(l,x) \right)\nu(dh) 
\\&\quad - \gamma V_{l}(l,x) + \mu V_x(l,x) + \frac{\sigma^2}2V_{xx}(l,x)  - \delta V(l,x)+\sup\limits_{u\in[0,\xi]}u\big\{e^{-l}-V_x(l,x)\big\} = 0
\end{split}
\end{equation*}
which by Equation \eqref{e:l factors} simplifies to
\begin{align}
\mu V_x(0,x) + \frac{\sigma^2}2V_{xx}(0,x)-\beta V(0,x)+\sup\limits_{u\in[0,\xi]}u\big\{1-V_x(0,x)\big\} = 0.\label{hjb:2}
\end{align}
Equation \eqref{hjb:2} has been studied by several authors and the solution is discussed in detail in \cite[pp. 97--104]{HS}. We summarise the implications in Proposition \ref{main statement} below.

In case of unrestricted dividend payments, the decision maker seeks to maximise $V^D$ for $D\in\mathcal A_u$, i.e.\ we want to find an admissible strategy $D^*\in\mathcal A_u$ such that
 $$ V^{D^*}(l,x) = \sup_{D\in\mathcal A_u} V^D(l,x). $$
As in Equation \eqref{e:l factors}, we get $V^D(l,x)=e^{-l}V^D(0,x)$ for any strategy $D\in\mathcal A_u$. The corresponding HJB reads
\begin{align*}
\max\Big\{&\frac A2 V_{ll}(l,x) + \int \left(V(l+h,x) - V(l,x) - h1_{\{|h|\leq 1\}}V_l(l,x) \right)\nu(dh) \\&\quad - \gamma V_{l}(l,x) + \mu V_x(l,x) + \frac{\sigma^2}2V_{xx}(l,x)  - \delta V(l,x), e^{-l} - V_x(l,x)\Big\} = 0
\end{align*}
and simplifies to
\begin{align}
 \max\Big\{&\mu V_x(0,x) + \frac{\sigma^2}2V_{xx}(0,x)-\beta V(0,x),1 - V_x(0,x)\Big\} = 0.\label{hjb:simple unrestricted}
\end{align}
Again, this equation has been treated extensively in \cite[p. 102]{HS}.

Define now the following auxiliary quantities
\begin{equation}
\begin{split}
&\theta:=\frac{-\mu+\sqrt{\mu^2+2\beta\sigma^2}}{\sigma^2}\quad\quad \mbox{and}\quad\quad \zeta:=\frac{-\mu-\sqrt{\mu^2+2\beta\sigma^2}}{\sigma^2}; \\&\eta:=\frac{\xi-\mu-\sqrt{(\xi-\mu)^2+2\beta\sigma^2}}{\sigma^2}\;.
\end{split}
\label{constants}
\end{equation}
\begin{Prop}\label{main statement}
In the restricted case, the value function $V$ solves HJB Equation \eqref{hjb:2} and is given by $V(l,x)=e^{-l}F(x)$, where
\begin{align*}
F(x):=\begin{cases}
\frac\xi\beta\big(1-e^{\eta x}\big) & \beta\ge-\xi\eta,
\\\frac{e^{\theta x}-e^{\zeta x}}{\theta e^{\theta  x_r}-\zeta e^{\zeta x_r}} & \mbox{$\beta<-\xi\eta$ and $x\le x_r$,}
\\\frac{\xi}{\beta}+\frac1\eta e^{\eta(x- x_r)} & \mbox{$\beta<-\xi\eta$ and $x> x_r$}\;,
\end{cases}
\end{align*}
and the optimal strategy $\bD^*=\{u^*_s\}$ is
\[
u^*_s(x)=\xi \one_{[X_s^{\bD^*}>x_r]}\;,
\]
with 
\begin{align*}
x_r:
=\frac{\ln\Big(1-\zeta\big(\frac\xi\delta+\frac{1}{\eta}\big)\Big)-\ln\Big(1-\theta\big(\frac\xi\delta+\frac{1}{\eta}\big)\Big)}{\theta-\zeta}=\frac{\ln\Big(\frac{\zeta^2-\eta\zeta}{\theta^2-\eta\theta}\Big)}{\theta-\zeta}\;.
\end{align*}
In the unrestricted case, the value function $V$ solves HJB Equation \eqref{hjb:simple unrestricted} and is given by $V(l,x)=e^{-l}G(x)$, where
\[
G(x):=\begin{cases}
\frac{e^{\theta x}-e^{\zeta x}}{\theta e^{\theta  x_u}-\zeta e^{\zeta  x_u}} & x\le  x_u
\\\frac{\mu}{\beta}+x- x_u & x> x_u
\end{cases},
\]
with the optimal strategy $D^*_t=\max\{\sup\limits_{0\le s\le t\w\tau^{D^*}} X_t-x_u,0\}$, and
\[
x_u:=\frac{\ln\Big(\frac{\zeta^2}{\theta^2}\Big)}{\theta-\zeta}\;.
\]
\end{Prop}
\begin{proof}
See \cite[pp. 101,103]{HS}.
\end{proof}
Let us make a simple example showing that the optimisation problem can be well-posed even in the presence of a negative preference rate $\delta$.
\begin{figure}[t]
\includegraphics[scale=0.6, bb = -50 280 200 500]{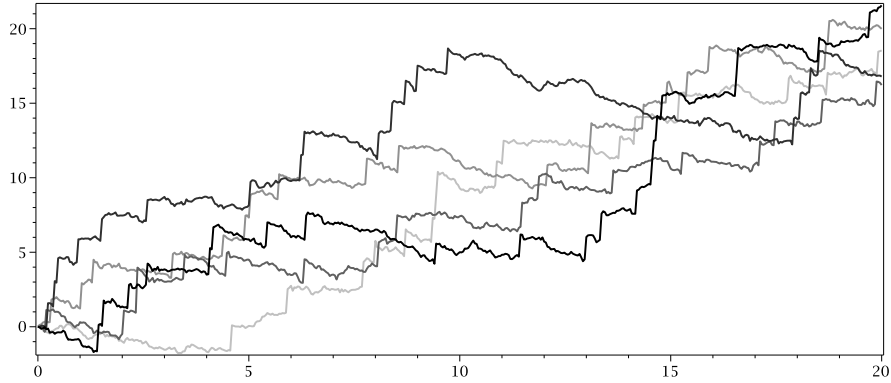}
\caption{Realisations of the process $\{L_t+\delta t\}$ with a negative preference rate $\delta$.\label{fig.1}}
\end{figure}
\begin{Bsp}
Let $\delta:=-1/4$, $A=1/2$, $b_1:=\log(5)$, $b_2:=\log(5/4)$, $N^{(1)},N^{(2)}$ be independent standard Poisson processes, i.e.\ they are Poisson processes with jump-intensity $1$ and jump-height $1$, and let $B$ be a standard Brownian motion independent of $W$, where $B,N^{(1)}$ and $N^{(2)}$ are $\mathcal F$-adapted. Define
$$L_t := \sqrt{A}\ B_t + b_1N^{(1)}_t - b_2N^{(2)}_t,\quad t\geq0.$$
Then, $\beta = 1/20$ and hence Proposition \ref{main statement} yields that the optimisation problem is well-posed. Sample paths simulations suggest that the paths of $L_t+\delta t$ tend to $\infty$ as $t\rightarrow\infty$ as one expects and which can be easily verified, cf. Figure \ref{fig.1}.
\end{Bsp}
An other example is to use a normal inverse Gaussian (NIG) process for the exchange rate. 
\begin{Bsp}
  Let $L$ be an NIG process with parameters $(\sigma^2,\theta,\kappa)=(0.19,0,1)$ in the sense of \cite[Table  4.5]{RamaCpeterT} and assume that $\delta=0.5$. Then $\beta = \delta + \frac{1}{\kappa}\left(1-\sqrt{1-\sigma^2\kappa+2\theta\kappa}\right) = 0.6 > 0$ and hence Proposition \ref{main statement} yields that the optimisation problem is well posed. Sample paths simulations for $\{L_t+\delta t\}$ can be found in Figure \ref{fig.2}.
\end{Bsp}
\begin{Bem}
Let now the constants defined in \eqref{constants} and correspondingly the optimal barriers $x_r$ and $x_u$ be functions of $\beta$, where we substitute $\beta$ by a non-negative variable $y$. It is straightforward to show that $x_r(y)$ and $x_u(y)$ are strictly decreasing in $y$. 
In particular, it means that the dividend barrier should be adjusted to the expectations of the insurance company concerning the future development of the exchange rate.
\end{Bem}
\begin{figure}[t]
\includegraphics[scale=0.45,angle = -90, bb = 100 -40 550 350]{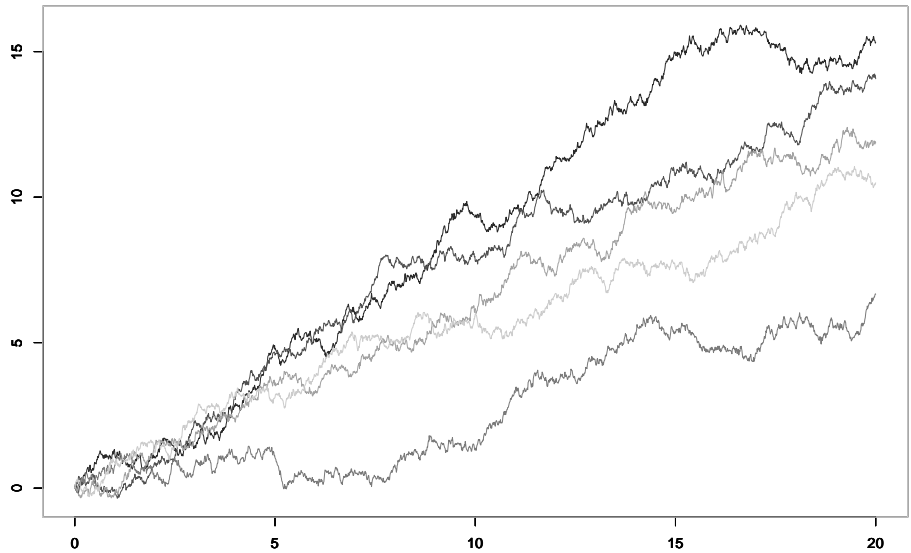}
\caption{Realisations of the process $\{L_t+\delta t\}$ with $\delta=0.5$ and an NIG process $L$ with parameters $(0.19,0,1)$.\label{fig.2}}
\end{figure}

\begin{Bem}[Incorporating dependencies between $L$ and $W$.]\label{r:depndency}
 It seems to be a natural extension to allow for dependencies between the L\'evy process $L$ and the Brownian motion $W$. 
 
However, the only way to introduce dependencies between a L\'evy process $L$ and a Brownian motion $W$, adapted to the same filtration, is to let the continuous martingale part $L^c$ of the L\'evy process $L$ depend on the Brownian motion $W$. The remainder $L^r:=L-L^c$ is a L\'evy process again, independent of $(L^c,W)$, and $L^c$ is a scaled Brownian motion. A natural way would be to correlate $L^c$ and $W$ so that $(W,L^c)$ is a two-dimensional Brownian motion. Then, $L^r$ can be removed from the consideration by adjusting $\delta$ as before. The remaining problem has been solved in \cite{shreve} and \cite{astak}.
\end{Bem}
\subsection*{Acknowledgements}
\noindent
The research of the first author was supported by the Austrian Science Fund, grant P26487.

\end{document}